\documentclass[10pt,a4paper,oneside]{article}
\usepackage{setspace,graphicx,amsmath,amsfonts,amsgen,amstext,amsthm,amsbsy,amsopn,amssymb,txfonts,bbm,pstricks,tikz,parskip,float,relsize}  
\usepackage[square,numbers]{natbib} 
\usepackage[utf8]{inputenc}        

\let\oldsqrt\sqrt
\def\sqrt{\mathpalette\DHLhksqrt}
\def\DHLhksqrt#1#2{%
\setbox0=\hbox{$#1\oldsqrt{#2\,}$}\dimen0=\ht0
\advance\dimen0-0.2\ht0
\setbox2=\hbox{\vrule height\ht0 depth -\dimen0}%
{\box0\lower0.4pt\box2}}

\newcommand{\easiest}{\setlength{\unitlength}{3mm}
\begin{picture}(0,0)(4.7,-0.2)
\put(4.75,0.75){$\scriptscriptstyle{x_0}$}
\put(4.75,-0.75){$\scriptscriptstyle{x_1}$}
\put(6.5,0.75){$\scriptscriptstyle{y_0}$}
\put(6.5,-0.75){$\scriptscriptstyle{y_1}$}
\put(5,-0.4){\line(0,1){1}}
\put(6.75,-0.4){\line(0,1){1}}
\put(5.25,0.1){\line(1,0){0.525}}
\put(5.975,0.1){\line(1,0){0.525}}
\put(5.775,-0.2){\line(0,1){0.6}}
\put(5.975,-0.2){\line(0,1){0.6}}
\end{picture}
}

\newcommand{\easy}{\setlength{\unitlength}{3mm}
\begin{picture}(0,0)(0,0)
\put(0,0){$\scriptsize{x_0}$}
\put(0.9,0.15){\line(1,0){1}}
\put(2.05,0){$\scriptsize{y_0}$}
\end{picture}
}

\newcommand{\easyij}{\setlength{\unitlength}{3mm}
\begin{picture}(0,0)(0,0)
\put(0,0){$\scriptsize{x_i}$}
\put(0.9,0.15){\line(1,0){1}}
\put(2.05,0){$\scriptsize{y_j}$}
\end{picture}
}

\newcommand{\xoeasy}{\setlength{\unitlength}{3mm}
\begin{picture}(0,0)(0,0)
\put(0,0){$\scriptsize{0_0}$}
\put(0.9,0.15){\line(1,0){1}}
\put(2.05,0){$\scriptsize{x_0}$}
\end{picture}
}

\newcommand{\zoeasy}{\setlength{\unitlength}{3mm}
\begin{picture}(0,0)(0,0)
\put(0,0){$\scriptsize{0_0}$}
\put(0.9,0.15){\line(1,0){1}}
\put(2.05,0){$\scriptsize{e_{1_0}}$}
\end{picture}
}

\newcommand{\hard}{\setlength{\unitlength}{3mm}
\begin{picture}(0,0)(12.8,-0.2)
\put(12.75,0.75){$\scriptscriptstyle{x_0}$}
\put(12.75,-0.75){$\scriptscriptstyle{x_1}$}
\put(14.5,0.75){$\scriptscriptstyle{y_0}$}
\put(14.5,-0.75){$\scriptscriptstyle{y_1}$}
\put(13.4,0.875){\line(1,0){1}}
\put(14.4,-0.625){\line(-1,0){1}}
\put(13.25,-0.525){\line(1,1){1.2}}
\put(14.55,-0.425){\line(-1,1){1.2}}
\put(13,-0.4){\line(0,1){1}}
\put(14.75,0.6){\line(0,-1){1}}
\end{picture}
}

\newcommand{\hardest}{\setlength{\unitlength}{3mm}
\begin{picture}(0,0)(8.8,-0.2)
\put(8.75,0.75){$\scriptscriptstyle{x_0}$}
\put(8.75,-0.75){$\scriptscriptstyle{x_1}$}
\put(10.5,0.75){$\scriptscriptstyle{y_0}$}
\put(10.5,-0.75){$\scriptscriptstyle{y_1}$}
\put(9.5,0.875){\line(1,0){1}}
\put(9.5,-0.625){\line(1,0){1}}
\put(10,-0.5){\line(0,1){0.5}}
\put(9.7,0){\line(1,0){0.6}}
\put(9.7,0.2){\line(1,0){0.6}}
\put(10,0.2){\line(0,1){0.5}}
\end{picture}
}

\newcommand{\constconn}{\setlength{\unitlength}{3mm}
\begin{picture}(0,0)(0,0)
\put(0,0){$\scriptsize{x_0}$}
\put(0.9,0.15){\line(1,0){1}}
\put(2.05,0){$\scriptsize{x_1}$}
\end{picture}
}

\newcommand{\dirconno}{\setlength{\unitlength}{3mm}
\begin{picture}(0,0)(12.8,-0.2)
\put(12.75,0.75){$\scriptscriptstyle{x_0}$}
\put(12.75,-0.75){$\scriptscriptstyle{x_1}$}
\put(14.5,0.75){$\scriptscriptstyle{y_0}$}
\put(14.5,-0.75){$\scriptscriptstyle{y_1}$}
\put(13.5,0.875){\vector(1,0){1}}
\put(14.4,-0.625){\vector(-1,0){1}}
\put(13,-0.4){\vector(0,1){1}}
\put(14.75,0.6){\vector(0,-1){1}}
\end{picture}
}

\newcommand{\dirconnt}{\setlength{\unitlength}{3mm}
\begin{picture}(0,0)(12.8,-0.2)
\put(12.75,0.75){$\scriptscriptstyle{x_0}$}
\put(12.75,-0.75){$\scriptscriptstyle{x_1}$}
\put(14.5,0.75){$\scriptscriptstyle{y_0}$}
\put(14.5,-0.75){$\scriptscriptstyle{y_1}$}
\put(13.25,-0.525){\vector(1,1){1.2}}
\put(14.55,-0.525){\vector(-1,1){1.2}}
\put(13,0.6){\vector(0,-1){1}}
\put(14.75,0.6){\vector(0,-1){1}}
\end{picture}
}

\newcommand{\dirconnth}{\setlength{\unitlength}{3mm}
\begin{picture}(0,0)(12.8,-0.2)
\put(12.75,0.75){$\scriptscriptstyle{x_0}$}
\put(12.75,-0.75){$\scriptscriptstyle{x_1}$}
\put(14.5,0.75){$\scriptscriptstyle{y_0}$}
\put(14.5,-0.75){$\scriptscriptstyle{y_1}$}
\put(14.4,0.875){\vector(-1,0){1}}
\put(14.4,-0.625){\vector(-1,0){1}}
\put(13.25,-0.525){\vector(1,1){1.2}}
\put(13.25,0.575){\vector(1,-1){1.2}}
\end{picture}
}

\newcommand{\spinconfig}[9]{\setlength{\unitlength}{#9cm}

\put(0,0.95){\line(1,0){1.6}}
\put(0,1.05){\line(1,0){1.6}}

\put(-0.1,1.2){$#1$}
\put(0.3,1.2){$#2$}
\put(1.0,1.2){$#3$}
\put(1.4,1.2){$#4$}

\put(-0.1,0.7){$#5$}
\put(0.3,0.7){$#6$}
\put(1.0,0.7){$#7$}
\put(1.4,0.7){$#8$}

}

\theoremstyle{definition}

\theoremstyle{remark}

\theoremstyle{plain}   
\numberwithin{equation}{section}
\newtheorem{partition}{Proposition}[section]
\newtheorem{factors}[partition]{Lemma}
\newtheorem{correlations}[partition]{Proposition}
\newtheorem{correlationinequalities}[partition]{Corollary}
\newtheorem{main}[partition]{Theorem}

\newtheorem*{Neel}{Theorem}

\begin{document}

\title{\Large{Existence of N\'eel order in the S=1 bilinear-biquadratic Heisenberg model via random loops}}
\author{Benjamin Lees\footnote{b.lees@warwick.ac.uk}\\
\small{Department of Mathematics, University of Warwick,
Coventry, CV4 7AL, United Kingdom}
}
\date{}
\maketitle

\begin{abstract}
\noindent
We consider the general spin-1 SU(2) invariant Heisenberg model with a two-body interaction. A random loop model is introduced and relations to quantum spin systems is proved. Using this relation it is shown that for dimensions 3 and above N\'eel order occurs for a large range of values of the relative strength of the bilinear ($-J_1$) and biquadratic ($-J_2$) interaction terms. The proof uses the method of reflection positivity and infrared bounds. Links between spin correlations and loop correlations are proved.
\end{abstract}

\section{Introduction}
\label{sec Intro}
\subsection{Historical Setting}

In this work properties of the spin-1 Heisenberg model are deduced using a random loop model first introduced in the work of Nachtergaele \cite{N}. Random loop models have been around since the work of T\'oth \cite{T} and Aizenman and Nachtergaele \cite{A-N}. In \cite{T} a lower bound was obtained on the pressure of the spin-$\frac{1}{2}$ Heisenberg ferromagnet; this improved the bound of Conlon and Solovej \cite{C-S}. Sharp bounds have recently been found \cite{C-G-S}. The loop model presented in \cite{A-N} applies to the spin-$\frac{1}{2}$ Heisenberg antiferromagnet. Both spin models can be applied to higher spins; for a review of these models we refer, for example, to \cite{G-U-W}. The work of Ueltschi \cite{U} combined and extended  these loop models. It has recently seen attention for its usefulness in several aspects of quantum spin systems. In \cite{U} it is shown that there is long-range order in various spin systems, including nematic order in the spin-1 system. The work of Crawford, Ng and Starr \cite{C-N-S} on emptiness formation also makes use of the model, as does the work of  Bj\"ornberg and Ueltschi \cite{B-U} on decay of correlations in the presence of a transverse magnetic field. The loop model presented here comes from \cite{N}, it is similar in flavour to the Aizenman-Nachtergaele-T\'oth-Ueltschi representation. See Refs. \cite{A-N,T,U,N,N2} and references therein.

Quantum spin systems are currently a very active area of research. The growth of popularity of probabilistic representations has allowed new methods to be applied to these systems with many interesting results. This work looks at the general SU(2) invariant spin-1 Heisenberg model with a two-body interaction
\begin{equation}
 H^{J_1,J_2}_{\Lambda}=-\sum_{\{x,y\}\in\mathcal{E}}\left(J_1\left(\mathbf{S}_x\cdot\mathbf{S}_y\right)+J_2\left(\mathbf{S}_x\cdot\mathbf{S}_y\right)^2\right).
\end{equation}
Here we will have $x\in\Lambda\subset\mathbb{Z}^d$ and $\mathcal{E}$ the set of nearest neighbour edges. The operators $\mathbf{S}=(S^1,S^2,S^3)$ are the spin-1 matrices, see section \ref{sec spin1model} for details of the model. The work in \cite{U} shows that in the region $0\leq J_1\leq \frac{1}{2}J_2$ the system exhibits \emph{nematic order} in the thermodynamic limit if the temperature is low enough and the dimension is high enough. Nematic order was also shown independently using different methods in \cite{T-T-I}. It is also shown that if $\Lambda$ is bipartite there will be \emph{N\'eel} order for $J_1=0\leq J_2$ at low temperature. This corresponds to the occurrence of infinite loops in the related loop model. Alternatively in $d\leq2$ infinite loops should not occur, it is proved in \cite{F-U} that this is the case for $J_2=0$, the extension to $J_2>0$ should be straightforward. The first proof of continuous symmetry breaking was shown by Fr\"ohlich, Simon and Spencer \cite{F-S-S} for the classical Heisenberg ferromagnet (and hence antiferromagnet). This result was extended by Dyson, Lieb and Simon \cite{D-L-S} to the quantum antiferromagnet. The result excluded the case $d=3$ and $S=\frac{1}{2}$, it was extended to this case in the work of Kennedy, Lieb and Shastry \cite{K-L-S}. These works all used the method of reflection positivity and infrared bounds. For information on reflection positivity see Refs. \cite{B,B-C-S,F-I-L-S,F-I-L-S2} and references therein. The Heisenberg ferromagnet is not reflection positive and hence does not benefit from these methods.

\subsection{Main result}

In this article we use the method of reflection positivity and infrared bounds on a random loop model. Links between correlations in the spin model and probabilities of events in the loop model are also derived in section 5. We focus on the quadrant $J_1\leq 0\leq J_2$, see Fig. \ref{phasediagram} (in fact we need only work on the unit circle $J_1^2+J_2^2=1$ but the quadrant is more convenient pictorially). The following result concerning N\'eel order follows from Proposition \ref{spincorrprop} (a), Theorem \ref{thmneel} and the discussion that follows. For the precise statements see Section \ref{sec neelorder}.
\\
\begin{Neel}
For $\Lambda\subset \mathbb{Z}^d$ a box of even side length and $d\geq 3$ there exists $\alpha=\alpha(d)>0$ such that for $J_1\leq0\leq J_2$ if $-J_1/J_2 >\alpha$ then there exists $c=c(\alpha,d)>0$ such that

\begin{equation}
 \lim_{\beta\to\infty}\lim_{|\Lambda|\to\infty}\frac{1}{|\Lambda|}\sum_{x\in\Lambda}(-1)^{\|x\|}\langle S_0^3S_x^3\rangle_{\Lambda,\beta}\geq c.
\end{equation}
Furthermore $\alpha(d)\to0$ as $d\to\infty$.
\end{Neel}

Note that the ``liminf'' version of this theorem will be proved. The limits exist but this will not be proved here.
It is shown in the discussion after Theorem \ref{thmneel} that this sum is positive if
\begin{equation}
 I_dK_d<(-4J_1)/(-J_1+4J_2).
\end{equation}
$I_d$ and $K_d$ are integrals to be introduced in \eqref{integrals}. Their values for various $d$ are given in the table below.
\begin{center}
\setlength{\unitlength}{4mm}
\begin{figure}[t!]

\includegraphics[width=10cm]{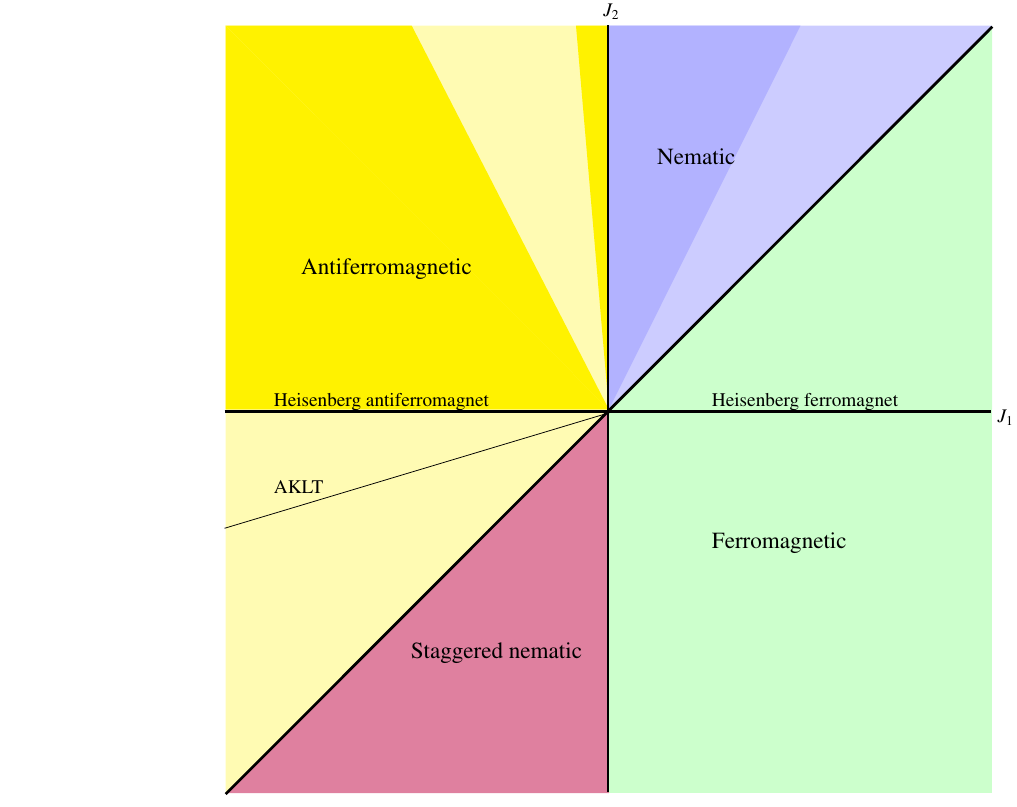}

\caption{\footnotesize{The phase diagram for the general SU(2) invariant spin-1 model. Regions that are shaded darker have rigorous proofs of the relevant phases. The line $J_1<0$, $J_2=0$ is the Heisenberg antiferromagnet where antiferromagnetic order has been proven \cite{D-L-S}, N\'eel order extends into the dark yellow region. The dark blue region $0\leq J_1\leq \frac{1}{2}J_2$ has nematic order at low temperatures \cite{U}, with N\'eel order on the line $J_2>0$, $J_1=0$. The adjacent dark yellow region has been proved to exhibit nematic order in high enough dimension \cite{L}. Antiferromagnet order is expected here but is not yet proved.}}
\label{phasediagram}
\end{figure}
\end{center}
\small
\begin{center}
\begin{tabular}{ l l l }
\hline
\hline
  $d$ & $I_d$ & $K_d$ \\
  3 & 0.349882 & 1.15672 \\
  4 & 0.253950 & 1.09441 \\
  5 & 0.206878 & 1.06754 \\
  6 & 0.177716 & 1.05274 \\
\hline
\hline
\end{tabular}
\end{center}
\normalsize
It can be shown \cite{D-L-S,K-L-S} that $I_d\to0$ and $K_d\to1$ as $d\to\infty$ and that both are decreasing in $d$. This means we can prove that the region where N\'eel order occurs will increase to the entire quadrant $J_1\leq0\leq J_2$ as $d\to\infty$ i.e., the ratio $\alpha(d)$ is decreasing.  In $d=3$ there is N\'eel order at low temperature in the spin system for $-J_1/J_2<0.46$, this is a triangular region of angle $65^\circ$ measured from the $J_1$ axis.

%In order to extend into the region $J_2>0$ we use the Falk-Bruch inequality \cite{F-B}. This means we must control the double %commutator terms coming the biquadratic ($J_2$) interaction, this is achieved by appealing to the random loop model, it is %certainly not clear how to handle these terms directly.

Reflection positivity for this quadrant is already known; for $J_1<0=J_2$  it was shown in \cite{D-L-S} and for $J_1=0<J_2$ one can see, for example, \cite{L} lemma 3.4 for an explicit proof. It was proved in \cite{D-L-S} that N\'eel order occurs for $J_1<0=J_2$, it is clear the result extends to a neighbourhood of the axis $J_1<0<J_2$ with $J_2$ sufficiently small. However it is impossible to extend the result concerning N\'eel order any significant amount without some new results. This is where the loop model has been essential. Indeed in \cite{D-L-S} an infrared bound is obtained of the form
\begin{equation}\label{DLSIRB}
 \widehat{(S_0^3,S_x^3)}_{Duh}(k)\leq \frac{1}{2(-J_1)\varepsilon(k)}
\end{equation}
where $(A,B)_{Duh}$ is the Duhamel correlation function and $\varepsilon(k)=2\sum_{i=1}^d(1-\cos k_i)$ for $k\in\Lambda^*$. Notice that this bound becomes weaker as $|J_1|$ decreases (equivalently on the unit circle as $|J_1|/|J_2|$ decreases). Transferring this bound to $\widehat{\langle S_0^3S_x^3\rangle}_\beta(k)$ requires the Falk-Bruch inequality, which would involve dealing with the term $\langle [\hat{S}_{-k}^3,[J_2\left(\mathbf{S}_x\cdot\mathbf{S}_y\right)^2,\hat{S}_k^3]]\rangle_\beta$. After some calculation one obtains correlations in Proposition \ref{spincorrprop} such as $\langle S_x^1S_y^1S_x^3S_y^3\rangle_\beta$. Hence to work directly in the quantum system using the methods of \cite{D-L-S} one must obtain good bounds on these correlations. Simple bounds such as taking the operator norm are not sufficient due to the weakening of \eqref{DLSIRB} as $|J_1|/|J_2|$ decreases. Without using the loop model it is not clear how to obtain such bounds currently.

The random loop model is presented in sections \ref{sec loopmodel} and \ref{sec spinconfigs}. The spin-1 Heisenberg model is introduced in Section \ref{sec spin1model}. In Section \ref{sec looprep} the connection between the loop model and the quantum system is proved. In particular it is shown how to write various correlation functions in terms of probabilities of events in the loop model; some of these correlations are also presented in \cite{U2}. In Section \ref{sec neelorder} the main result concerning N\'eel order is presented and proved.

\section{The random loop model}
\label{sec loopmodel}

We now introduce the loop model presented in \cite{N}. To begin we take a finite set of vertices, $\Lambda$, with a set of edges, $\mathcal{E}\subset\{\{x,y\}|x,y\in\Lambda,x\neq y\}$. We associate to this lattice a new lattice, $\tilde\Lambda$, and edge set, $\tilde{\mathcal{E}}$:
\begin{align}
\tilde{\Lambda}=&\Lambda\times\{0,1\},
\\
\tilde{\mathcal{E}}=&\{\{(x,i),(y,j)\}|i,j\in\{0,1\}, \{x,y\}\in\mathcal{E}\}.
\end{align}
There are two lattice sites in $\tilde\Lambda$ for every site in $\Lambda$ and four edge in $\tilde{\mathcal{E}}$ for each edge in $\mathcal{E}$. We will write $x_0,x_1$ in place of $(x,0),(x,1)$.
\par
For $\beta>0$ consider a process, $\rho_{J_1,J_2}$, consisting of a Poisson point process on $\mathcal{E}\times[0,\beta]$ and a uniform measure on segments of $\Lambda\times[0,\beta]$ between events of the Poisson point process. The Poisson point process has two events that we will refer to as `single bars' and `double bars'. Note that this process is on the edge set $\mathcal{E}$, the events define corresponding events on the edge set $\tilde{\mathcal{E}}$. The single bars will occur at rate $-2J_1$ and double bars at rate $J_2$ for $J_1\leq 0\leq J_2$. The rate for the single bars is written in this way to be consistent with the connection to the quantum spin system that will be introduced in Section \ref{sec spin1model}. The interval $[0,\beta]$ will be referred to as a time interval. The uniform measure is on two possibilities, ``crossing'' and ``parallel''. How to build loops from these events is described in detail below, see Fig. \ref{eventtypes} for pictorial representations of the events.
\begin{figure}[t]
 \includegraphics[width=12cm]{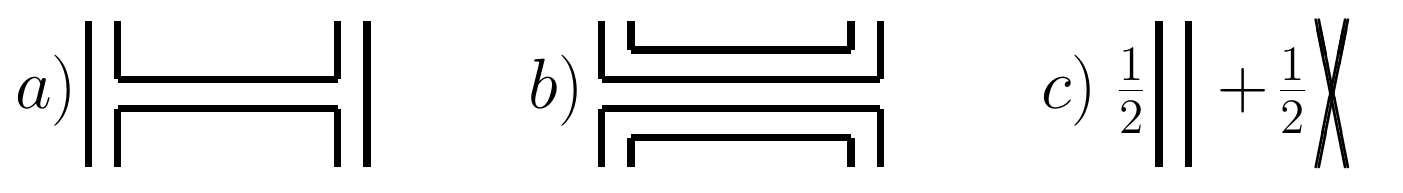}
 \caption{Events of the process $\rho_{J_1,J_2}$, a) represents single bars, b) represents double bars and c) represents the uniform measure on vertical segments being either parallel or crossing.}
 \label{eventtypes}
\end{figure}

\par
We first define the single and double bars. Single bars occur at a point $(x,y,t)$ for $\{x,y\}\in\mathcal{E}$. We define the corresponding geometric event on $\tilde{\mathcal{E}}$ as a bar joining $x_1$ and $y_0$ at time $t$. Double bars occur at a point $(x,y,t)$ and the corresponding event on $\tilde{\mathcal{E}}$ is a bar joining $x_1$ and $y_0$ and a bar joining $x_0$ and $y_1$, both at time $t$. A \emph{loop} of length $l$ is then a map $\gamma:[0,\beta l]_{per}\to\tilde{\Lambda}\times[0,\beta]_{per}$ such that $\gamma(s)\neq\gamma(t)$ if $s\neq t$, $\gamma$ is piecewise differentiable with derivate $\pm 1$ where it exists. If $s$ is a point of non-differentiability then $\{\gamma(s-),\gamma(s+)\}\in\tilde{\mathcal{E}}$. Loops with the same support and different parameterisations are identified. For a realisation $\omega$ of $\rho_{J_1,J_2}$ we associate a set of loops as follows: Starting at a point $(x_i,s)\in\tilde{\Lambda}\times[0,\beta]$ we move upwards (i.e. in direction of increasing $s$). If a bar is met at time $t$ it is crossed and we then continue in the opposite direction from $(y_j,t)$, where $y_j$ is the other site associated to the bar. Each maximal vertical segment between bars $(x_0,x_1)\times[s,t]$ (i.e. bars involving the site $x$ occur at times $s$ and $t$ and no bar involving $x$ occurs for $u$ such that $s<u<t$) is either parallel (nothing happens) or crossing (the sites $x_0$ and $x_1$ are exchanged). If time $\beta$ is reached the periodic time conditions mean we continue in the same direction starting from time 0. We denote by $\mathcal{L}(\omega)$ the set of all loops associated to a realisation $\omega$. Loops are most easily understood pictorially, see Fig. \ref{loopexample}.
\begin{figure}[t]
\includegraphics[width=12cm]{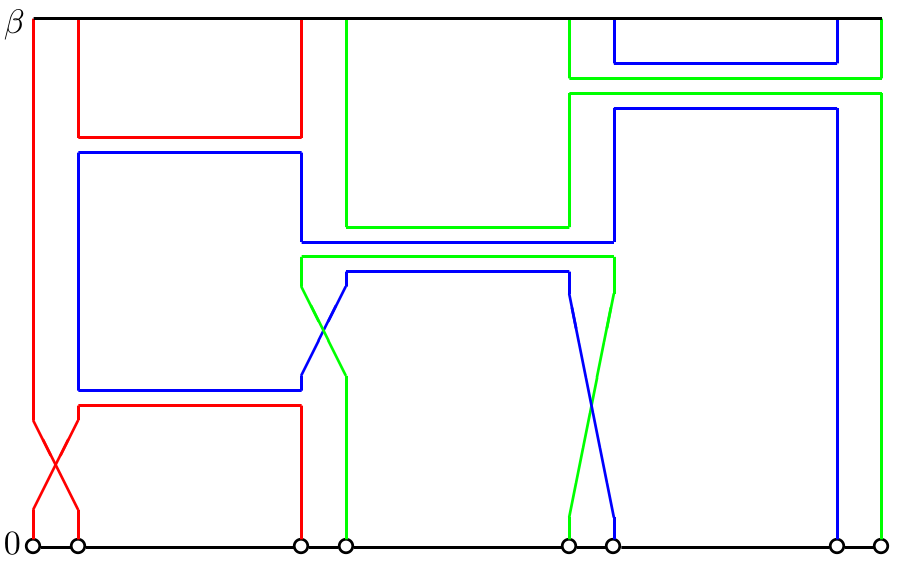}
\caption{A example realisation with loops coloured differently, here there are four sites in the underlying $\Lambda$ and for this realisation $|\mathcal{L}(\omega)|=3$.}
\label{loopexample}
\end{figure}
Note that the loops could be defined via a Poisson point process on $\tilde{\mathcal{E}}\times[0,\beta]$ where bars can occur between $x_i$ and $y_j$ with each $(i,j)$ being equally likely. However one would still need to introduce the crossing or parallel events so that it is still possible to have $x_0$ and $x_1$ in the same loop even when there is no bar occurring on any edge containing $x$.
 
For this loop model we have partition function
\begin{equation}
 Y_\theta^{J_1,J_2}(\beta,\Lambda)=\int\rho_{J_1,J_2}(\mathrm{d}\omega)\theta^{|\mathcal{L}(\omega)|}.
\end{equation}
Here $\theta>0$ is a parameter and $\rho_{J_1,J_2}$ is the probability measure corresponding to a Poisson point process of intensity $-2J_1$ for single bars and $J_2$ for double bars. The relevant probability measure is then
\begin{equation}
\frac{1}{Y_\theta^{J_1,J_2}(\beta,\Lambda)}\rho_{J_1,J_2}(\mathrm{d}\omega)\theta^{|\mathcal{L}(\omega)|}.
\end{equation}
We are interested in sets of realisations, $\omega$, where certain points of $\tilde{\Lambda}\times[0,\beta]$ are in the same loop. Probabilities of these events are connected to correlations in the spin-1 quantum system presented in Section \ref{sec spin1model}, they will be required in the proof of N\'eel order in Section \ref{sec neelorder}. Particular events of interest will be denoted pictorially , see Fig. \ref{connectiontypes}. These events are defined and denoted as follows.
\begin{description}
\item[a)] The event that sites $x_i$ and $y_j$ are connected (in the same loop). Note that the probability of $x_i$ and $y_j$ being connected is independent of $i$ and $j$. Denoted $E[\easyij\qquad\;\,]$.
\item[b)] The event that $x_0$ and $x_1$ are connected, $y_0$ and $y_1$ are connected but there is no connection from any $x_i$ to any $y_j$. Denoted $E\bigg[\easiest\qquad\bigg]$.
\item[c)] The event that $x_0$ and $y_0$ are connected, $x_1$ and $y_1$ are connected but $x_0$ and $x_1$ are not connected. Denoted $E\bigg[\hardest\qquad\bigg]$. We can also have $x_0$ and $y_1$ connected and $x_1$ and $y_0$ connected but $x_0$ and $x_1$ not connected and denote the event in the analogous way. These events both have the same probability.
\item[d)] The event that all four sites $x_0,x_1,y_0,y_1$ are connected. Denoted $E\bigg[\hard\qquad\bigg]$.
\end{description}
The definition of bars means that if a loop is followed starting from a point $x_i\in\tilde\Lambda$ (by moving in either the up or down direction) then the direction it is travelling upon arriving at a point $y_j\in\tilde\Lambda$ in the same loop is determined only by the number of bars the loop has encountered between the sites. For example on a bipartite lattice defined by sublattices $\Lambda_A$ and $\Lambda_B$ such that $\{x,y\}\in\mathcal{E} \iff x\in\Lambda_A,y\in\Lambda_B$ the direction that $x_i$ is left and $y_j$ is entered will be the same if $x$ and $y$ are in the same sublattice and different if they are in different sublattices.
\par
Sometimes the order in which sites are encountered along the loop will be important. In this case arrows will indicate the order that sites will be encountered in on following the loop (up to parameterisation). The events $E\bigg[\dirconno\qquad\bigg]$\,, $E\bigg[\dirconnt\qquad\bigg]$\, and $E\bigg[\dirconnth\qquad\bigg]$\, are the events that all four sites are connected and are encountered along the loop in the order indicated by the arrows. For example the first event,$E\bigg[\dirconno\qquad\bigg]$, means that upon leaving site $x_0$ if we encounter $y_0$ before encountering $x_1$ then we will then encounter $y_1$ and then $x_1$ before closing the loop. As this notation is potentially confusing (but also seemingly unavoidable) the reader will be told explicitly when the order is important. When wanting the probability of these events we will drop the $E$ from the notation, as below.
\par
It is intuitively clear that $\mathbb{P}(\easy\qquad\;\,)$ decays exponentially fast with respect to $\|x-y\|$ for $\beta$ small. Hence $\mathbb{P}\bigg(\hard\qquad\bigg)$ and $\mathbb{P}\bigg(\hardest\qquad\bigg)$ must also have exponential decay. $\mathbb{P}\bigg(\easiest\qquad\bigg)$ should depend weakly on $\|x-y\|$ for small enough $\beta$. For $\|x-y\|$ large enough the probability may approach $\mathbb{P}(\constconn\qquad\;\;)^2$, it is not clear how to prove or disprove such a relation at this time.
 
\begin{figure}[t!]
\includegraphics[width=12cm]{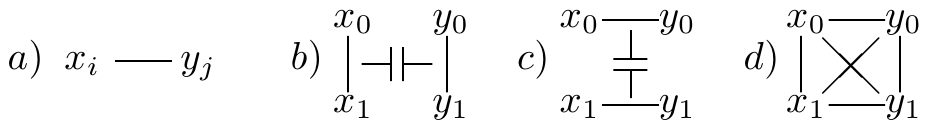}
\caption{Pictures representing the set of realisations where the pictured connections are present.
}
\label{connectiontypes}
\end{figure}
 
\section{Space-time spin configurations}
\label{sec spinconfigs}
In order to make the connection with spin systems we need the notion of a \emph{space-time spin configuration}. The spin system we shall connect to is the spin-1 Heisenberg model, we shall make this connection via an intertwining that merges two spin-$\frac{1}{2}$ models. For this reason we will take $\theta=2$ from Section \ref{sec loopmodel} ($2S+1$ for $S=\frac{1}{2}$). This is also the reason the lattice $\tilde{\Lambda}$ has two sites for every site in $\Lambda$. It is also possible to represent the spin-$S$ model for general $S$ by merging $2S$ spin-$\frac{1}{2}$ models, this will mean $\tilde{\Lambda}$ will have $2S$ sites for every site in $\Lambda$. See \cite{N} for more details. This generalisation together with some results analogous to the ones presented here should be straightforward once the spin-1 model is understood. It is not immediately clear which results will still hold however, investigation is required.
\par
From now on we take the cubic lattice in $\mathbb{Z}^d$ with side length L, denoted $\Lambda_L$, with periodic boundary conditions. The edge set, $\mathcal{E}_L$, will consist of pairs of nearest neighbour lattice points. Precisely
\begin{align}
\Lambda_L =&\left\{-\frac{L}{2}+1,\dots,\frac{L}{2}\right\}^d,
\\
\mathcal{E}_L=&\{\{x,y\}\subset\Lambda_L|\; \|x-y\|=1\ \text{ or } |x_i-y_i|=L-1 \text{ for some } i=1,...,d\}.
\end{align}
Where $\|x-y\|$ is the graph distance between $x$ and $y$.
A space-time spin configuration is a function
\begin{equation}
 \sigma:\tilde{\Lambda}\times[0,\beta]_{per}\to\left\{-\frac{1}{2},\frac{1}{2}\right\}.
\end{equation}
$\sigma_{x_i,t}$ is piecewise constant in $t$ for any $x_i$. We further define $\Sigma$ to be the set of all such functions with a finite number of discontinuities. For a realisation of the process $\omega$ we consider $\sigma$ that are constant on the vertical segments of each loop in $\mathcal{L}(\omega)$ and that change value on crossing a bar. This restriction on configurations will allow to make the link with spin systems. We call such configurations \emph{compatible} with $\omega$ and denote by $\Sigma^{(1)}(\omega)$ the set of all compatible configurations.
The following relation holds as we work on a bipartite lattice, meaning fixing a configuration's value at some $(x_i,t)$ determines the configuration on the entire loop containing $(x_i,t)$:
\begin{equation}
|\Sigma^{(1)}(\omega)|=2^{|\mathcal{L}(\omega)|},
\end{equation}
from which we can obtain
\begin{equation}
 Y_2^{J_1,J_2}(\beta,\Lambda)=\int\rho_{J_1,J_2}(\mathrm{d}\omega)\sum_{\sigma\in\Sigma^{(1)}(\omega)}1.
\end{equation}
\par
We further define the set $\Sigma^{(1)}_{x_i,y_j}(\omega)\supset \Sigma^{(1)}(\omega)$ to be compatible configurations along with configurations that flip spin at points $(x_i,0)$ or $(y_j,0)$ (or both) but are otherwise compatible. 
\par
In the remainder of this section it will be convenient to ignore the condition that compatible configurations flip value on crossing a bars as it adds unnecessary extra complication. Hence we further define $\Sigma^{(2)}(\omega)$ to be configurations that are constant on loops (and hence do not flip value at bars). $\Sigma^{(2)}_{x_i,y_j}(\omega)\supset \Sigma^{(2)}(\omega)$ denotes the set of configurations in $\Sigma^{(2)}(\omega)$ along with configurations that flip spin at points $(x_i,0)$ or $(y_j,0)$ (or both) but are otherwise consistent with the definition of $\Sigma^{(2)}(\omega)$. The reader should bear in mind in the sequel that the connection with the spin-1 quantum system requires spin flips at bars, with this in mind the modifications to the remainder of this section required to incorporate the spin flips are easy.
\par
As in \cite{U} we will later need a more general setting for the measure on space-time spin configurations. We consider a Poisson point process on $\tilde{\mathcal{E}}\times[0,\beta]$ with events being specifications of the local spin configuration. We will consider discontinuities involving two pairs of sites ($x_0,x_1,y_0,y_1$). The objects of the process will be a set of allowed configurations at these sites immediately before and after $t$. We can denote these events as 
\par

\begin{equation}
\qquad\qquad\qquad\qquad\qquad\qquad\qquad
\end{equation}

\setlength{\unitlength}{1.5cm}
\begin{figure}[h!]
\begin{picture}(6,1.2)(-3,-0.8)
\spinconfig{\mathsmaller{\sigma_{x_0,t+}}}{\mathsmaller{\sigma_{x_1,t+}}}{\mathsmaller{\sigma_{y_0,t+}}}{\mathsmaller{\sigma_{y_1,t+}}}{\mathsmaller{\sigma_{x_0,t-}}}{\mathsmaller{\sigma_{x_1,t-}}}{\mathsmaller{\sigma_{y_0,t-}}}{\mathsmaller{\sigma_{y_1,t-}}}{1.5}
\end{picture}
\vspace{-2cm}
\end{figure}

Implicit here is an ordering on $\Lambda$ with $x<y$. An event $A$ is a subset of $\{-1/2,1/2\}^8$ and occurs with intensity $\iota(A)$. More precisely we let $\iota:\mathcal{P}(\{-1/2,1/2\}^8)\to\mathbb{R}$ denote the intensities of the Poisson point process, denoted $\rho_\iota$. Given realisation, $\xi$, of $\rho_\iota$ let $\Sigma(\xi)$ be the set of configurations compatible with $\xi$ meaning that $\sigma\in\Sigma(\xi)$ if\\

\setlength{\unitlength}{1.5cm}
\begin{figure}[h!]
\hspace{-4.4cm}
\begin{picture}(5,1)(-3,0)
\spinconfig{\mathsmaller{\sigma_{x_0,t+}}}{\mathsmaller{\sigma_{x_1,t+}}}{\mathsmaller{\sigma_{y_0,t+}}}{\mathsmaller{\sigma_{y_1,t+}}}{\mathsmaller{\sigma_{x_0,t-}}}{\mathsmaller{\sigma_{x_1,t-}}}{\mathsmaller{\sigma_{y_0,t-}}}{\mathsmaller{\sigma_{y_1,t-}}}{1.5}
\end{picture}
\vspace{-1.2cm}
\end{figure}
\vspace{-1.67cm}
\[
\hspace{2.5cm}
 \in A \text{ whenever } \xi \text{ contains the event } A \text{ at point } (x_0,x_1,y_0,y_1,t),
\]
and $\sigma_{x_i,t}$ is otherwise constant in $t$. The measure is then given by $\rho_\iota$ with the counting measure on compatible configurations. We note that different intensities can give the same measure as in \cite{U}, for $\iota$ and $\iota'$ intensities it is shown in \cite{U} that
\begin{equation}
\int\rho_\iota(\mathrm{d}\xi)\int\rho_{\iota'}(\mathrm{d}\xi')\sum_{\sigma\in\Sigma(\xi\cup\xi')}F(\sigma)=\int\rho_{\iota+\iota'}(\mathrm{d}\xi)\sum_{\sigma\in\Sigma(\xi)}F(\sigma).
\end{equation}

We want to write the Poisson point process involving bars in terms of intensities of specifications of spins. We require that specifications corresponding to single and double sets of bars have intensity $-2J_1$ and $J_2$ respectively. If we naively define $\tilde{\iota}$ by

\begin{picture}(0,0)(-1.75,1)
\spinconfig{a'}{a}{a}{b}{a'}{c}{c}{b}{1}
\end{picture}
\begin{picture}(0,0)(-4.75,1)
\spinconfig{a'}{a}{a}{a'}{c'}{c}{c}{c'}{1}
\end{picture}
\begin{equation}\label{naiveintensities}
\tilde{\iota}\,\Bigg(\Bigg\{\qquad\qquad\quad\Bigg\}\Bigg)=-2J_1,\qquad\tilde{\iota}\,\Bigg(\Bigg\{\qquad\qquad\quad \Bigg\}\Bigg)=J_2.
\end{equation}
For any $a,a',b,c,c'\in\{1/2,-1/2\}$, where the first event corresponds to single bars and the second event to double bars. We see there is an overlap on the specification\\
\begin{picture}(1,0.7)(-3.5,0.4)
\spinconfig{b}{a}{a}{b}{b}{c}{c}{b}{1}
\end{picture}

so this assignment of intensities of specifications cannot be correct. Simply removing the overlapping case from one of the specifications will result in events not having the required intensities. This suggests we should instead define $\iota$ by
\begin{picture}(0,0)(1.35,1.3)
\spinconfig{a'}{a}{a}{b}{a'}{c}{c}{b}{1}
\end{picture}
\begin{picture}(0,0)(4.55,1.3)
\spinconfig{a'}{a}{a}{a'}{c'}{c}{c}{c'}{1}
\end{picture}
\begin{equation}\label{intensities}
\iota\,\Bigg(\Bigg\{\qquad\qquad\quad\Bigg\}_{a'\neq b}\Bigg)=-2J_1,\quad \iota\,\Bigg(\Bigg\{\qquad\qquad\quad \Bigg\}_{a'\neq c'}\Bigg)=J_2,
\end{equation}
\begin{picture}(0,0)(-3,1.05)
\spinconfig{b}{a}{a}{b}{b}{c}{c}{b}{1}
\end{picture}
\[
\iota\,\Bigg(\Bigg\{\qquad\qquad\quad \Bigg\}\Bigg)=J_2-2J_1.
\]

\hspace{-0.13cm}
For any $a,a',b,c,c'\in\{1/2,-1/2\}$.
Now each specification is disjoint from the other two and single and double sets of bars have intensities $-2J_1$ and $J_2$ respectively, as required. We also have $\iota(A)=0$ for any other specification. Then
\begin{equation}
Y_2^{J_1,J_2}(\beta,\Lambda)=\int\rho_\iota(\mathrm{d}\xi)\sum_{\sigma\in\Sigma(\xi)}1.
\end{equation}
This representation can be used to show reflection positivity of the loop model, however in light of the famous paper of Dyson, Lieb and Simon \cite{D-L-S} showing reflection positivity for $J_1<0=J_2$ and the recent paper \cite{L} showing reflection positivity for $J_1=0<J_2$ we have reflection positivity for $J_1<0<J_2$ more directly, as we shall see.
 
\section{The general spin-1 SU(2) invariant Heisenberg model}
\label{sec spin1model}
Let $S\in\frac{1}{2}\mathbb{N}$. A spin-$S$ model has local Hilbert spaces $\mathcal{H}_x=\mathbb{C}^{2S+1}$. Observables are Hermitian matrices built from linear combinations of tensor products of operators on $\otimes_{x\in\Lambda}\mathcal{H}_x$. Physically important observables are expressed in terms of \emph{spin matrices} $S^1,S^2$ and $S^3$, operators on $\mathbb{C}^{2S+1}$ that are the generators of a (2$S$+1)-dimensional irreducible unitary representation of $\mathfrak{su}(2)$. They satisfy the commutation relations
\begin{equation}\label{spins}
 \left[S^\alpha,S^\beta\right]=i\sum_\gamma\mathcal{E}_{\alpha\beta\gamma}S^\gamma,
\end{equation}
where $\alpha,\beta,\gamma\in\{1,2,3\}$ and $\mathcal{E}_{\alpha\beta\gamma}$ is the Levi-Civita symbol. Denote $\mathbf{S}=(S^1,S^2,S^3)$, then $\mathbf{S}\cdot\mathbf{S}=S(S+1)\mathbbm{1}$. The case $S=\frac{1}{2}$ gives the Pauli spin matrices. For $S=1$ there are several choices for spin matrices, we will use the following matrices:
\begin{equation}\label{S1matrices}
S^1=\frac{1}{\sqrt{2}}\left(\begin{matrix} 0 & 1 & 0 \\  1 & 0 & 1 \\ 0 & 1 & 0  \end{matrix}\right),\quad S^2=\frac{1}{\sqrt{2}}\left(\begin{matrix} 0 & -i & 0 \\  i & 0 & -i \\  0 & i & 0  \end{matrix}\right),\quad S^3=\left(\begin{matrix} 1 & 0 & 0 \\ 0 & 0 & 0 \\ 0 & 0 & -1 \end{matrix}\right). 
\end{equation}

Consider a pair $(\Lambda,\mathcal{E})$ of a lattice, $\Lambda\subset\mathbb{Z}^d$, and a set of edges, $\mathcal{E}$, between points in $\Lambda$. We will take $\Lambda$ to be a box in $\mathbb{Z}^d$, hence $\Lambda$ is bipartite. We denote by $\Lambda_A$ and $\Lambda_B$ the two disjoint lattices such that $\Lambda_A\cup\Lambda_B=\Lambda$ and every $e\in\mathcal{E}$ contains precisely one site from $\Lambda_A$ and one site from $\Lambda_B$.

Then we take the operator $S_x^i$ for $i=1,2,3$ to be shorthand for the operator $S^i_x\otimes Id_{\Lambda\setminus\{x\}}$. Recall the definition of $\tilde{\Lambda}$ and $\tilde{\mathcal{E}}$ above, we shall use these below.
\par
The most general SU(2) invariant Hamiltonian with two-body interactions for spin-1 is
\begin{equation}\label{ham}
 H^{J_1,J_2}_{\Lambda}=-\sum_{\{x,y\}\in\mathcal{E}}\left(J_1\left(\mathbf{S}_x\cdot\mathbf{S}_y\right)+J_2\left(\mathbf{S}_x\cdot\mathbf{S}_y\right)^2\right).
\end{equation}
We will soon drop the parameters $J_1,J_2$ from $ H^{J_1,J_2}_{\Lambda}$ for readability. In this article we will be concerned with the region where $J_1\leq0\leq J_2$. Associated to this Hamiltonian we have the following partition function and Gibbs states for $\beta>0$:
\begin{align}
 Z_{\Lambda,\beta}^{J_1,J_2}=&Tr e^{-\beta H^{J_1,J_2}_{\Lambda}},
 \\
 \langle\cdot\rangle_{\Lambda,\beta}^{J_1,J_2}=&\frac{1}{Z_{\Lambda,\beta}^{J_1,J_2}}Tr\cdot e^{-\beta H^{J_1,J_2}_{\Lambda}}.
\end{align}
Again we shall drop the parameters $J_1,J_2$ from the notation.
\par
The following new definitions come from Nachtergaele \cite{N}. We introduce an isometry $V:\mathbb{C}^3\to\mathbb{C}^2\otimes\mathbb{C}^2$ with the property $VD^1(g)=(D^{\frac{1}{2}}(g))^{\otimes 2}V$ for $g\in SU(2)$ and $D^S$ the spin-$S$ representation of SU(2). Here the representation $D^1$ is given by the matrices \eqref{S1matrices} and $D^{\frac{1}{2}}$ is given by the Pauli matrices. it is clear such an isometry exists as we can define it for spin matrices and then extend by linearity (recall that the spin matrices generate the representation). From this we obtain the key relation
\begin{equation}
 VS^i=(\sigma^i\otimes\mathbbm{1}+\mathbbm{1}\otimes\sigma^i)V,
\end{equation}
where $\sigma^i$ are the spin-$\frac{1}{2}$ matrices (hence $2\sigma^i$ are the Pauli matrices). Further we have
\begin{equation}
 V^*V=\mathbbm{1} \text{ and } VV^*=P,
\end{equation}
where $P$ is the projection onto the spin triplet. Hence $VS^i$ acts on $\mathbb{C}^2\otimes\mathbb{C}^2$ and so using the notation before $V_xS_x^i$ acts on $\otimes_{x\in\Lambda}\mathbb{C}^2\otimes\mathbb{C}^2$. We make the following definition
\begin{equation}\label{R}
 R^i:=VS^iV^*.
\end{equation}
One can check that $R^i=(\sigma^i\otimes\mathbbm{1}+\mathbbm{1}\otimes\sigma^i)$.
To make expressions more concise we will also denote $A_X:=\otimes_{x\in X}A_x$ for $X\subset\Lambda$. For these new operators we have a new Hamiltonian  (note we have now dropped the $J_1$ and $J_2$ parameters)
\begin{equation}
 \tilde{H}^{(1)}_{\tilde{\Lambda}}=-\sum_{\{x,y\}\in\mathcal{E}}\left(J_1\left(\mathbf{R}_x\cdot\mathbf{R}_y\right)+J_2\left(\mathbf{R}_x\cdot\mathbf{R}_y\right)^2\right),
 \end{equation}
 and associated Gibbs states
 \begin{align}
 Z_{\tilde{\Lambda},\beta}^{(1)}=&TrP_\Lambda e^{-\beta  \tilde{H}^{(1)}_{\tilde{\Lambda}}},
 \\
 \langle\cdot\rangle_{\tilde\Lambda,\beta}^{(1)}=&\frac{1}{Z_{\tilde{\Lambda},\beta}^\sim} Tr\cdot P_{\Lambda}e^{-\beta  \tilde{H}^{(1)}_{\tilde{\Lambda}}}.
\end{align}
The connection with the previous Gibbs state can easily be made explicit,
\begin{equation}\label{newgibbs}
 \langle A \rangle_{\Lambda,\beta}=\langle V_\Lambda A V_\Lambda^*\rangle^{(1)}_{\tilde{\Lambda},\beta}.
\end{equation}
We use Dirac notation in the following way: $|a,b\rangle$ denotes an element of the one site Hilbert space $\mathbb{C}^2\otimes\mathbb{C}^2$ and $|a,b\rangle\otimes|c,d\rangle$ for two sites etc.

There are two operators of particular interest, both act on two sites. Firstly we define $\mathcal{S}^{(1)'}$ by its matrix elements
\begin{equation}\label{calS1}
 \langle a',b'|\otimes\langle c',d'|\mathcal{S}^{(1)'}|a,b\rangle\otimes|c,d\rangle=(-1)^{b-b'}\delta_{a,a'}\delta_{d,d'}\delta_{b,-c}\delta_{b',-c'}.
\end{equation}
Geometrically this requires spin $b$ and $c$ and the spins $b'$ and $c'$ to be the negative of each other and also requires $a=a'$ and $d=d'$. This corresponds to the the single bars in the loop picture. The second operator, $\mathcal{D}^{(1)'}$, is also defined via its matrix elements
\begin{equation}\label{calD1}
 \langle a',b'|\otimes\langle c',d'|\mathcal{D}^{(1)'}|a,b\rangle\otimes|c,d\rangle=(-1)^{a-a'}(-1)^{b-b'}\delta_{a,-d}\delta_{b,-c}\delta_{a',-d'}\delta_{b',-c'}.
\end{equation}
The geometrical interpretation this time is that of the double bars. The actual operators needed are $\mathcal{S}^{(1)}=P\mathcal{S}^{(1)'}P$ and $\mathcal{D}^{(1)}=P\mathcal{D}^{(1)'}P$ in order to account for bars occurring between any $x_i$ and $y_j$ with each $i$ and $j$ from $\{0,1\}$ being equally likely. Note here that from this definition we see that we require the spin value to change sign on crossing a bar as was mentioned in Section \ref{sec spinconfigs}. There are also extra factors in $\mathcal{S}^{(1)}$ and $\mathcal{D}^{(1)}$ of $e^{i\pi a}$ for the bottom half of a bar (denoted $\sqcap$) and $e^{-i\pi a}$ for the top half of a bar (denoted $\sqcup$) where $a=\pm\frac{1}{2}$ is the spin value on the  site in $\Lambda_A$ associated to the bar. By direct computation of the matrix elements we can prove the relations
\begin{align}
 \mathcal{S}^{(1)}_{x,y}=-\frac{1}{2}\mathbf{R}_x\cdot\mathbf{R}_y+\frac{1}{2}P_{x,y},
 \\
 \mathcal{D}^{(1)}_{x,y}=\left(\mathbf{R}_x\cdot\mathbf{R}_y\right)^2-P_{x,y}.
\end{align}
Using these relations we can rewrite the Hamiltonian in the region $J_1\leq0\leq J_2$ as
\begin{equation}\label{newham}
 \tilde{H}^{(1)}_{\tilde{\Lambda}}=-\sum_{\{x,y\}\in\mathcal{E}}\left(-2J_1 \mathcal{S}^{(1)}_{x,y}+J_2 \mathcal{D}^{(1)}_{x,y}+(J_1+J_2)P_{x,y}\right).
\end{equation}

\section{The random loop representation}
\label{sec looprep}
We can neglect the term $(J_1+J_2)P_{x,y}$ in the Hamiltonian \eqref{newham} and instead add $(2J_1-J_2)\mathbbm{1}$, this does not change the Gibbs states. Doing this allows to use a useful lemma from \cite{A-N}
\begin{equation}
 \exp\left\{-\sum_{\{x,y\}\in\mathcal{E}}\left(uA_{x,y}+\nu B_{x,y}-u-\nu\right)\right\}=\int\rho(\mathrm{d}\omega)\sideset{}{^*}\prod_{(x,y)\in\omega}C_{x,y}.
\end{equation}
Here $\rho$ is the measure associated to a Poisson point process on $\mathcal{E}\times[0,1]$ with two events occurring with intensities $u$ and $\nu$ respectively. The product is ordered according to the times at which the events occur. $C$ is either $A$ or $B$ depending on which event occurs. This is actually a slight extension of the lemma presented in \cite{A-N}. From this we can obtain
\begin{equation}\label{randlooprep}
 \exp\left\{-\sum_{\{x,y\}\in\mathcal{E}}\left(-2J_1\mathcal{S}^{(1)}_{x,y}+J_2\mathcal{D}^{(1)}_{x,y}+2J_1-J_2\right)\right\}=\int\rho(\mathrm{d}\omega)\sideset{}{^*}\prod_{(x_i,y_j)\in\omega}A^{(1)}_{x_i,y_j}
\end{equation}
here each $A^{(1)}$ is one of $\mathcal{S}^{(1)}$ or $\mathcal{D}^{(1)}$. The process has intensity $-2J_1$ for single bars and $J_2$ for double bars. Again the product is ordered by the time events occur.
\par
We now prove the connection between the loop model and the quantum system. This will enable us to understand certain important correlation functions. After this we should have the tools we need to calculate any two point correlation (at least ones involving only spin operators). The first thing to understand is the extra factor, which we shall denote by $z_{x_i,y_j}(\sigma,\omega)$, the product of all factors $e^{\pm i \pi a}$ from operators $\mathcal{S}^{(1)}$ and $\mathcal{D}^{(1)}$ corresponding to the bars in loop(s) containing $x_i$ and $y_j$ in a realisation $\omega$ of $\rho_{J_1,J_2}$. Again $a\in\{1/2,-1/2\}$ is the value that $\sigma$ assigns to the portions of these loop(s) in the $\Lambda_A$ sublattice (or if all of a loop is on the sublattice $\Lambda_B$ $a$ is given by the negative of the value assigned to the loop). The value of $z_{x_i,y_j}(\sigma,\omega)$ is specified by the following lemma:
\\
\begin{factors}\label{zfactor}
For $\Lambda$ bipartite we have for all $i,j$
\begin{equation}
z_{x_i,y_j}(\sigma,\omega)= \left\{
  \begin{array}{c l}
    1 & \text{ if } \sigma\in \Sigma^{(1)}(\omega)\\
    (-1)^{\|x-y\|} & \text{ if } \sigma\in\Sigma^{(1)}_{x_i,y_j}(\omega)\setminus\Sigma^{(1)}(\omega) \text{ and } \omega\in E[\easyij\qquad\;\,]. \\
  \end{array}
\right.
\end{equation}
\end{factors}
Before the proof we should note that the lemma says that the only dependence on $\sigma$ is at $x_i$ and $y_j$ at time zero. If the spin does not flip at both sites that we get total factor $1$, else it depends on which sublattices the sites are in. If the spin only flips at one site then there are no compatible configurations hence the value of the total extra factor is unimportant. 

\begin{proof}
To begin note that we can take $(i,j)=(0,0)$. The result for $(i,j)\neq(0,0)$ follows as the choice of $i$ or $j$ does not affect which sublattice the two sites are in. Suppose $\sigma\in \Sigma^{(1)}(\omega)$. Moving upwards from $x_0$ the first bar encountered is $\sqcap$, the bars encountered then alternate between $\sqcup$ and $\sqcap$. Moving downwards from $x_0$ we first encounter a bar $\sqcup$ then alternate between $\sqcap$ and $\sqcup$. This means we can make a matching between bars of the form $\sqcap$ and bars of the form $\sqcup$. Because there are no spin flips at time zero all the bars $\sqcap$ have factors $e^{i\pi a}$ and all the bars $\sqcup$ have factor $e^{-i\pi a}$ where $a$ is the spin value $\sigma$ gives to $x_0$ at time zero. Hence we have full cancellation and are left with factor $1$. If there were a spin flip then bars between $x_0$ at time $0-$ and $y_0$ at time $0\pm$ would have factors $e^{i\pi(-a)}$ and $e^{-i\pi(-a)}$ for $\sqcap$ and $\sqcup$ respectively.
\par
If $\sigma\in\Sigma^{(1)}_{x_0,y_0}(\omega)\setminus\Sigma^{(1)}(\omega)$ and  $(-1)^{\|x-y\|}=1$ and $\omega\in E[\easy\qquad\;\,]$, then $x_0$ and $y_0$ are in the same sublattice. We can thus deduce that the section of loop that moves upwards/downwards from $x_0$ crosses an even number of bars before reaching $y_0$. This means that the loop containing $x_0$ and $y_0$ contains an even number of bars of each type ($\sqcap$ or $\sqcup$). Hence we can make a matching of a bar $\sqcup$ in one `half' of the loop with a bar $\sqcup$ in the other `half' and the same with bars $\sqcap$, with some bars left over. The factors from bars in the matching will thus be $1$ as the spin flip at $x_0$ at time $0$ means one bar in each pair has factor $e^{\pm i \pi a}$ and one bar has factor $e^{\pm i \pi (-a)}$. Here by `half' of a loop we mean the section that connects $x_0$ at time $0+$ with $y_0$ at time $0-$ or $x_0$ at time $0-$ with $y_0$ at time $0+$. There are still possibly some bars left over as each half of the loop may have a different 
number of bars in it. A moments thought reveals that there must be an even number of bars left, half of type $\sqcap$ and half of type $\sqcup$. As the bars $\sqcap$ have factor $e^{-i\pi (\pm a)}$ and the bars $\sqcup$ have factor $e^{i\pi(\pm a)}$ we have full cancellation again and have total factor $1$.
\par
For the remaining case $\sigma\in\Sigma^{(1)}_{x_0,y_0}(\omega)\setminus\Sigma^{(1)}(\omega)$ and  $(-1)^{\|x-y\|}=-1$ and $\omega\in E[\easy\qquad\;\,]$, we have $x_0$ and $y_0$ in different sublattices. We can see as last time that the factors from the `extra bars' (that arise from each half of the loop having a different number of bars) will cancel as again there are equal numbers of $\sqcap$ and $\sqcup$. For the remaining bars there are an odd number in each half of the loop, this means we can make a matching for all but two of the bars. The factors from bars in the matching will cancel each other. For the remaining two bars one is a $\sqcap$ with factor $e^{i\pi(\pm a)}$ and one is a $\sqcup$ with factor $e^{-i\pi(\mp a)}$ (the sign of $a$ is opposite due to the spin flip at $x_0$ at time $0$). This means the overall factor is $(\pm i)^2=-1$. This completes the proof.
\end{proof}

 In light of Proposition \ref{zfactor} the following proposition can be proved in the same way as Theorem 3.2 in \cite{U}.
\\
\begin{partition}\label{loopspin prop} The partition functions $Z_{\tilde{\Lambda},\beta}^{(1)}$ are given by
 \begin{equation}
   Z_{\tilde{\Lambda},\beta}^{(1)}=\int\rho(\mathrm{d}\omega)\sum_{\Sigma^{(1)}(\omega)}\prod_{\{x_i,y_j\}\in\mathcal{E}}z_{x_i,y_j}(\sigma,\omega)
   =\int\rho(\mathrm{d}\omega)2^{|\mathcal{L}(\omega)|}=Y_2^{J_1,J_2}(\beta,\Lambda).
\end{equation}
We also have the following identity
\begin{equation}
 Tr (\sigma^3\otimes\mathbbm{1})_x(\sigma^3\otimes\mathbbm{1})_ye^{-\beta \tilde{H}^{(1)}_{\tilde{\Lambda}}}=\int\rho(\mathrm{d}\omega)\sum_{\Sigma^{(1)}(\omega)}\left(\prod_{\{x_i,y_j\}\in\mathcal{E}}z_{x_i,y_j}(\sigma,\omega)\right)\sigma_{x_0}\sigma_{y_0},
\end{equation}
where $\sigma_{z_i}$ is the value of a space time configuration, $\sigma$, at time $0$ and site $z_i$.
\end{partition}
 
With the important details understood we can calculate some correlations in terms of probabilities in the loop model. The most important correlations here are the N\'eel and nematic correlations (Proposition \ref{spincorrprop} (a) and (b) respectively).
\\
\begin{correlations}\label{spincorrprop} For $i,j=1,2,3$, $x\neq y$, $i\neq j$ and $\Lambda$ bipartite
\begin{description}
\item[(a)] $\langle S_x^iS_y^i\rangle_{\Lambda,\beta}=(-1)^{\|x-y\|}\mathbb{P}(\easy\qquad\;\,)$,
\item[(b)] $\langle (S_x^i)^2(S_y^i)^2\rangle_{\Lambda,\beta}-\langle (S_x^i)^2\rangle_{\Lambda,\beta}\langle (S_y^i)^2\rangle_{\Lambda,\beta}=-\frac{1}{36}+\frac{1}{4}\mathbb{P}\bigg(\easiest\qquad\bigg)+\frac{1}{2}\mathbb{P}\bigg(\hardest\qquad\bigg)+\frac{1}{4}\mathbb{P}\bigg(\hard\qquad\bigg)$,
\item[(c)] $\langle S_x^iS_x^jS_y^iS_y^j\rangle_{\Lambda,\beta}=\frac{1}{4}\left[-(-1)^{\|x-y\|}\mathbb{P}(\easy\qquad\;\,)+\mathbb{P}\bigg(\hardest\qquad\bigg)\right]$,
\item[(d)] $\langle S_x^iS_x^jS_y^jS_y^i\rangle_{\Lambda,\beta}=\frac{1}{4}\left[(-1)^{\|x-y\|}\mathbb{P}(\easy\qquad\;\,)+\mathbb{P}\bigg(\hardest\qquad\bigg)\right]$,
\item[(e)] $\langle(S_x^i)^2(S_y^j)^2\rangle_{\Lambda,\beta}=\frac{5}{12}+\frac{1}{4}\left[\mathbb{P}\bigg(\easiest\qquad\bigg)+\mathbb{P}\bigg(\dirconno\qquad\bigg)-\mathbb{P}\bigg(\dirconnth\qquad\bigg)\right]$.
\end{description}
\end{correlations}

\begin{proof}
We will calculate the correlations in order. First note that each $S^i$ plays an equivalent role, hence cyclic permutations of the indices $(1,2,3)$ does not alter the expectation. Using this together with $(S^iS^j)^T=\pm(S^jS^i)$ (the sign depending on the value of $i$ and $j$) means we can take $i=3$ and $j=1$. For each we will expand using \eqref{R} and \eqref{newgibbs}.
\par
\textbf{\emph{Proof of (a)}}. First
\begin{equation}
\langle S_x^3S_y^3\rangle_{\Lambda,\beta}=\langle (\sigma^3\otimes\mathbbm{1}\otimes\sigma^3\otimes\mathbbm{1}+\sigma^3\otimes\mathbbm{1}\otimes\mathbbm{1}\otimes\sigma^3+\mathbbm{1}\otimes\sigma^3\otimes\sigma^3\otimes\mathbbm{1}+\mathbbm{1}\otimes\sigma^3\otimes\mathbbm{1}\otimes\sigma^3)_{x,y}\rangle^{(1)}_{\tilde\Lambda,\beta}.
\end{equation}
We see that due to sites $z_0$ and $z_1$ being interchangeable for $z\in\Lambda$ each of the four terms in the sum have the same expectation. We also know from Proposition \ref{loopspin prop}
\begin{equation}
 Tr (\sigma^3\otimes\mathbbm{1})_x(\sigma^3\otimes\mathbbm{1})_ye^{-\beta \tilde{H}_{\tilde{\Lambda}}}=\int\rho(\mathrm{d}\omega)\sum_{\Sigma^{(1)}(\omega)}\sigma_{x_0}\sigma_{y_0}.
\end{equation}
We note that the integral differs from zero only on the set where $x_0$ and $y_0$ are connected. If $x$ and $y$ are in different sublattices the product of spin configuration values is $-\frac{1}{4}$, if in the same sublattice the product is $\frac{1}{4}$. We can deduce that 
\begin{equation}\label{33corr}
 \langle S^3_xS^3_y\rangle_{\Lambda,\beta}=(-1)^{\|x-y\|}\mathbb{P}(\easy\qquad\;\,).
\end{equation}
\par
\textbf{\emph{Proof of (b)}}. For the second correlation
\begin{equation}
(R_x^3)^2=(\sigma^3\otimes\mathbbm{1}+\mathbbm{1}\otimes\sigma^3)^2_x=\left(\frac{1}{2}\mathbbm{1}\otimes\mathbbm{1}+2\sigma^3\otimes\sigma^3\right)_x.
\end{equation}
We see that expanding as before gives 
\begin{equation}
\langle(S_x^3)^2\rangle_{\Lambda,\beta}=\langle(R_x^3)^2\rangle_{\tilde\Lambda,\beta}^{(1)}=\frac{1}{Z_{\tilde{\Lambda},\beta}^\sim}\int\rho(\mathrm{d}\omega)\sum_{\sigma\in\Sigma^{(1)}(\omega)}\left(\frac{1}{2}+2\sigma_{x_0}\sigma_{x_1}\right)=\frac{1}{2}+\frac{1}{2}\mathbb{P}(\constconn\qquad\;\;).
\end{equation}
From this and the fact that $\langle(S_x^3)^2\rangle_{\Lambda,\beta}=\frac{1}{3}\langle\mathbf{S}_x\cdot\mathbf{S}_x\rangle_{\Lambda,\beta}=\frac{2}{3}$ we can deduce that 
\begin{equation}\label{01prob}
\mathbb{P}(\constconn\qquad\;\;)=\frac{1}{3}.
\end{equation}
For the first term in the correlation we again note that $\langle (S_x^3)^2(S_y^3)^2\rangle_{\Lambda,\beta}=\langle (R_x^3)^2(R_y^3)^2\rangle_{\Lambda,\beta}^1$. We then calculate as before:
\begin{equation}\label{correxp}
\begin{aligned}
(R_x^3)^2(R_y^3)^2=&(\sigma^3\otimes\mathbbm{1}+\mathbbm{1}\otimes\sigma^3)^2_x(\sigma^3\otimes\mathbbm{1}+\mathbbm{1}\otimes\sigma^3)^2_y
\\
=&\left(\frac{1}{2}\mathbbm{1}\otimes\mathbbm{1}+2\sigma^3\otimes\sigma^3\right)_x\left(\frac{1}{2}\mathbbm{1}\otimes\mathbbm{1}+2\sigma^3\otimes\sigma^3\right)_y
\\
=&\left(\frac{1}{4}\mathbbm{1}^{\otimes 4}+\sigma^3\otimes\sigma^3\otimes\mathbbm{1}\otimes\mathbbm{1}+\mathbbm{1}\otimes\mathbbm{1}\otimes\sigma^3\otimes\sigma^3+4(\sigma^3)^{\otimes 4}\right)_{x,y}.
\end{aligned}
\end{equation}
Now following through the same expansion as before we have
\begin{equation}
\langle (S_x^3)^2(S_y^3)^2\rangle_{\Lambda,\beta}=\frac{1}{Z}\int\rho(\mathrm{d}\omega)\sum_{\sigma\in\Sigma^{(1)}(\omega)}\left(\frac{1}{4}+\sigma_{x_0}\sigma_{x_1}+\sigma_{y_0}\sigma_{y_1}+4\sigma_{x_0}\sigma_{x_1}\sigma_{y_0}\sigma_{y_1}\right).
\end{equation}
Using \eqref{01prob} and noting that the last term in the sum requires either two loops containing two of the sites $x_0,x_1,y_0,y_1$ each or one loop containing all four sites to give a non-zero contribution to the sum overall (if one site is not connected to any other its spin value can be $\pm\frac{1}{2}$ independently of other sites, averaging the integral on this set to zero) we have
\begin{equation}
\langle (S_x^3)^2(S_y^3)^2\rangle_{\Lambda,\beta}-\langle(S_x^3)^2\rangle_{\Lambda,\beta}\langle(S_y^3)^2\rangle_{\Lambda,\beta}=-\frac{1}{36}+\frac{1}{4}\mathbb{P}\bigg(\easiest\qquad\bigg)+\frac{1}{2}\mathbb{P}\bigg(\hardest\qquad\bigg)+\frac{1}{4}\mathbb{P}\bigg(\hard\qquad\bigg).
\end{equation}
The probability $\mathbb{P}\bigg(\hardest\qquad\bigg)$ comes with twice the weight because there are two ways to connect both sites at $x$ to different sites at $y$ (but only one way both sites at $x$ can be connected and both sites at $y$ can be connected).
\par
\textbf{\emph{Proof of (c)}}. For the third correlation we use the same expansion
\begin{equation}
\langle S_0^1S_0^3S_x^1S_x^3\rangle_{\Lambda,\beta}=\frac{4}{Z_{\tilde{\Lambda},\beta}^{(1)}}Tr (\sigma^1\sigma^3\otimes\mathbbm{1}+\sigma^1\otimes\sigma^3)_x(\sigma^1\sigma^3\otimes\mathbbm{1}+\sigma^1\otimes\sigma^3)_yP_\Lambda e^{-\beta\tilde{H}^{(1)}_{\tilde{\Lambda}}}.
\end{equation}
The factor $4$ has come from grouping together terms such as $\sigma^1\otimes\sigma^3$ and $\sigma^3\otimes\sigma^1$ that have the same expectation. A useful observation at this stage is that $\sigma^1\sigma^3=\frac{-i}{2}\sigma^2$. Calculating further and noting that the two cross terms in the above product have the same expectation we see
\begin{equation}
\langle S_0^1S_0^3S_x^1S_x^3\rangle_{\Lambda,\beta}=4\left\langle-\frac{1}{4}\sigma^2\otimes\mathbbm{1}\otimes\sigma^2\otimes\mathbbm{1}-i\sigma^2\otimes\mathbbm{1}\otimes\sigma^1\otimes\sigma^3+\sigma^1\otimes\sigma^3\otimes\sigma^1\otimes\sigma^3\right\rangle_{\tilde{\Lambda},\beta}^{(1)}.
\end{equation}
From the symmetric roles of $\sigma^i$ for $i=1,2,3$ and part a) we know the first term is $-\frac{(-1)^{\|x-y\|}}{4}\mathbb{P}(\easy\qquad\;\,)$. For the second term we need $\langle\sigma^2\otimes\mathbbm{1}\otimes\sigma^1\otimes\sigma^3\rangle_{\tilde{\Lambda},\beta}^{(1)}$. This is the expectation of a matrix with purely imaginary entries, due to the one appearance of $\sigma^2$. Now we note three pieces of information that allow us to calculate this expectation. All the matrices $e^{-\beta\tilde{H}^{(1)}_{\tilde{\Lambda}}}, P_\Lambda, \sigma^1, \sigma^2$ and $\sigma^3$ are Hermitian. The matrices $\sigma^i$ are acting on different sites in $\tilde{\Lambda}$ and hence they commute. $e^{-\beta\tilde{H}^{(1)}_{\tilde{\Lambda}}}$ and $P_\Lambda$ commute and have real entries. This means taking the adjoint of the operator leaves the expectation unchanged. Because the operator is purely imaginary we should obtain the negative of what we started with on taking the 
adjoint. Hence the correlation must be zero.  
\par
For the last term we expand as in Proposition \ref{loopspin prop} and obtain
\begin{equation}
\langle\sigma^1\otimes\sigma^3\otimes\sigma^1\otimes\sigma^3\rangle_{\tilde{\Lambda},\beta}^{(1)}=\frac{1}{Z_{\tilde{\Lambda},\beta}^{(1)}}\int\rho(\mathrm{d}\omega)\sum_{\sigma\in\Sigma^{(1)}_{x_0,y_0}(\omega)}z_{x_0,y_0}(\sigma,\omega)\langle \sigma_{\cdot,0+}|\sigma^1\otimes\sigma^3\otimes\sigma^1\otimes\sigma^3|\sigma_{\cdot,0-}\rangle
\end{equation}
Here $\sigma_{\cdot,0\pm}$ denotes the full spin configuration for some $\sigma\in\Sigma_{x_0,y_0}(\omega)$ at time $0\pm$ respectively. Also note that as $\sigma^1$ flips spins and $\sigma^3$ does not the set of space-time spin configurations $\Sigma^{(1)}_{x_0,y_0}(\omega)$ is the correct set. We could expand the set of configurations we sum over to include configurations that flip spin at sites $x_1$ and $y_1$ at time zero but these would not be compatible with $\sigma^3$ acting at time zero at those sites hence they would not contribute. Recall that a loop that contains a site that spin flips at time zero cannot contain only one such site, hence the set of configurations that contribute to the integral is $E[\easy\qquad\;\,]$. Again the set of configurations where one of the sites $x_1$ or $y_1$ is not connected to any of the other three does not contribute to the integral. Combining these two facts we see that the only sets of configurations that contribute to the integral are those where there are two loops each containing two sites (one with $x_0$ and $y_0$ and the other with $x_1$ and $y_1$), or one loop containing all four sites. For the case of two loops there is one factor of $z_{x_0,y_0}(\sigma,\omega)=(-1)^{\|x-y\|}$ from the loop containing $x_0$ and $y_0$ (where $\sigma^1$ acts). Another factor of $(-1)^{\|x-y\|}$ comes from the loop containing $x_1$ and $y_1$ and the condition that the spin flips on crossing a bar. Note that for the first loop there is no such factor coming from spin flips at bars because $\sigma^1|\pm \frac{1}{2}\rangle=+\frac{1}{2}|\mp \frac{1}{2}\rangle$ hence there is a factor of $+\frac{1}{2}$ regardless of the spin value at the site. For the case of one loop containing all sites the order that sites occur in 
the loop is important; this is because both $\sigma^1$ and $\sigma^3$ are acting at sites in the loop. If, when following the loop, the site $y_1$ appears directly before or after the site $x_1$ then the section of loop between these sites follows the normal rule of flipping spins at bars (or if we follow the loop the other way we pass through two spin flips at time zero as well, these cancel each other out as far as the product of spins at sites $x_1$ and $y_1$ is concerned). This means we have a factor of $(-1)^{\|x-y\|}$ as before. If one of the sites $x_0$ or $y_0$ appears between sites $x_1$ and $y_1$ on the loop the effect of the extra spin flip changes the sign of the factor coming from the product of spins, giving a factor of $-(-1)^{\|x-y\|}$. As before we also have the factor $z_{x_0,y_0}(\sigma,\omega)=(-1)^{\|x-y\|}$ in both cases.
This means the correlation is
\begin{equation}
\langle\sigma^1\otimes\sigma^3\otimes\sigma^1\otimes\sigma^3\rangle_{\tilde{\Lambda},\beta}^{(1)}=\frac{1}{16}\left[\mathbb{P}\bigg(\hardest\qquad\bigg)+\mathbb{P}\bigg(\dirconno\qquad\bigg)-\mathbb{P}\bigg(\dirconnt\qquad\bigg)\right].
\end{equation}
Recall that the arrows in the events show the direction that the loop is traversed. From this we can finally deduce that
\begin{equation}
\langle S_x^1S_x^3S_y^1S_y^3\rangle_{\Lambda,\beta}=\frac{1}{4}\left[-(-1)^{\|x-y\|}\mathbb{P}(\easy\qquad\;\,)+\mathbb{P}\bigg(\hardest\qquad\bigg)+\mathbb{P}\bigg(\dirconno\qquad\bigg)-\mathbb{P}\bigg(\dirconnt\qquad\bigg)\right].
\end{equation}
Now we note that the last two probabilities are equal (swap $y_0$ and $y_1$).
\par
The correlations (d) and (e) follow easily using the same techniques and considerations as above.
\end{proof}
From this we can easily obtain some bounds on these correlations that are potentially very difficult to obtain without the loop model.
\\
\begin{correlationinequalities}
For $i,j=1,2,3$, $x\neq y$, $i\neq j$ and $\Lambda$ bipartite
\begin{description}
\item[(a)] $\langle (S_x^i)^2(S_y^i)^2\rangle_{\Lambda,\beta}-\langle (S_x^3)^2\rangle_{\Lambda,\beta}\langle (S_y^3)^2\rangle_{\Lambda,\beta}\leq \frac{1}{18}+\frac{3}{4}(-1)^{\|x-y\|}\langle S_x^iS_y^i\rangle_{\Lambda,\beta}$

\item[(b)] $\langle S_x^iS_x^jS_y^iS_y^j\rangle_{\Lambda,\beta}\leq \frac{1}{4}((-1)^{\|x-y\|}-1)\langle S_x^iS_y^i\rangle_{\Lambda,\beta}$

\item[(c)] $\langle S_x^iS_x^jS_y^jS_y^j\rangle_{\Lambda,\beta}\leq \frac{1}{4}((-1)^{\|x-y\|}+1)\langle S_x^iS_y^i\rangle_{\Lambda,\beta}$

\item[(d)] $\langle S_x^iS_x^jS_y^iS_y^j\rangle_{\Lambda,\beta} \left\{
  \begin{array}{l l}
    \geq 0 & \quad \text{if $\|x-y\|$ is odd}\\
    \leq 0 & \quad \text{if $\|x-y\|$ is even}
  \end{array} \right.$
  
\item[(e)] $\langle S_x^iS_x^jS_y^jS_y^i\rangle_{\Lambda,\beta} \left\{
  \begin{array}{l l}
    \leq 0 & \quad \text{if $\|x-y\|$ is odd}\\
    \geq 0 & \quad \text{if $\|x-y\|$ is even}
  \end{array} \right.$
\end{description}
\end{correlationinequalities}

\begin{proof}
All inequalities are immediate from Proposition \ref{spincorrprop} when we note that $E\bigg[\easiest\qquad\bigg]$ is a sub-event of $E[\constconn\qquad\;\,]$ and $E\bigg[\hardest\qquad\bigg]$ is a sub-event of $E[\easy\qquad\;\,]$.
\end{proof}
Other inequalities of interest involve correlations between nearest neighbour points. Equation (29) in \cite{T-T-I} allows us to obtain the following bound in the ground state ($\beta\to\infty$)
\begin{equation}
 \mathbb{P}(\zoeasy\qquad\quad)\geq \frac{1}{d}\frac{2J_2-3J_1}{4J_2-3J_1}.
\end{equation}
Now looking at Proposition \ref{spincorrprop} (b) for $\|x-y\|=1$ (say $x=0$, $y=e_1$) we see that if $J_1=0$ then the event $\mathbb{P}(\zoeasy\qquad\quad)$ puts us into the case of one of the last two probabilities. Ignoring the first probability (as it is difficult to control) we obtain (for $J_2>0=J_1$)
\begin{equation}
 \langle (S_0^3)^2(S_{e_1}^3)^2\rangle_{\Lambda,\beta}-\langle (S_0^3)^2\rangle_{\Lambda,\beta}\langle (S_{e_1}^
 3)^2\rangle_{\Lambda,\beta}\geq -\frac{1}{36}+\frac{1}{8d}.
\end{equation}
This bound is positive for $d\leq 4$, however it was not sufficient to deduce nematic order from a theorem analogous to \ref{thmneel} but concerning the nematic correlation function. 

\section{Occurrence of N\'eel Order}
\label{sec neelorder}
\subsection{Setting and results}

We take the cubic lattice in $\mathbb{Z}^d$ with side length L, denoted $\Lambda_L$, with periodic boundary conditions. The edge set, $\mathcal{E}_L$, will consist of pairs of nearest neighbour lattice points. Precisely
\begin{align}
\Lambda_L =&\left\{-\frac{L}{2}+1,\dots,\frac{L}{2}\right\}^d,
\\
\mathcal{E}_L=&\{\{x,y\}\subset\Lambda_L|\; \|x-y\|=1\ \text{ or } |x_i-y_i|=L-1 \text{ for some } i=1,...,d\}.
\end{align}
For the main theorem we need to introduce two integrals, they come about due to similar considerations as in \cite{K-L-S}
\begin{align}\label{integrals}
I_d=&\frac{1}{(2\pi)^d}\int_{[-\pi,\pi]^d}\left(\frac{1}{d}\sum_{i=1}^d\cos k_i\right)_+\sqrt{\frac{\varepsilon(k+\pi)}{\varepsilon(k)}}\mathrm{d}k,
\\
K_d=&\frac{1}{(2\pi)^d}\int_{[-\pi,\pi]^d}\sqrt{\frac{\varepsilon(k+\pi)}{\varepsilon(k)}}\mathrm{d}k.
\end{align}
Here $(\cdot)_+$ denotes the positive part and $\varepsilon(k)=2\sum_{i=1}^d(1-\cos(k_i))$.
\\
\begin{main}
\label{thmneel}
Let $d\geq 3$ and $J_1\leq 0\leq J_2$, for $L$ even we have the two bounds
\begin{align}
\lim_{\beta\to\infty}\lim_{L\to\infty}\frac{1}{|\Lambda_L|}\sum_{x\in\Lambda_L}(-1)^{\|x\|}\langle S_0^3S_x^3\rangle^{J_1,J_2}_{\Lambda_L,\beta}\geq
\begin{cases}
\sqrt{\mathbb{P}(\zoeasy\qquad\quad)}\left(\sqrt{\mathbb{P}(\zoeasy\qquad\quad)}-I_d\sqrt{\frac{1}{4}-\frac{J_2}{J_1}}\right),
\\
 1-K_d\sqrt{\mathbb{P}(\zoeasy\qquad\quad)}\sqrt{\frac{1}{4}-\frac{J_2}{J_1}}.
\end{cases}
\end{align}
\end{main}
Positivity of this lower bound implies N\'eel order for those values of $J_1$ and $J_2$ in the spin-1 system at low enough temperature. This result is also equivalent to the occurrence of macroscopic loops in the loop model.

Of course we see that for $-J_1,J_2>0$ the positivity of the lower bound doesn't depend on the value of $J_1^2+J_2^2$, only on the ratio $-J_1/J_2$. This means there corresponds an angle, measured from the $J_1$ axis, such that for angles less than this we have proved the existence of N\'eel order/macroscopic loops. The bound is positive if
\begin{equation}
 \sqrt{\mathbb{P}(\zoeasy\qquad\quad)}<\frac{1}{K_d}\sqrt{\frac{-4J_1}{-J_1+4J_2}} \text{ or } \sqrt{\mathbb{P}(\zoeasy\qquad\quad)}>I_d\sqrt{\frac{1}{4}-\frac{J_2}{J_1}}.
\end{equation}
One of these is certainly satisfied if $I_dK_d<(-4J_1)/(-J_1+4J_2)$. A table of values of $I_d$ and $K_d$ for various $d$ is presented in \cite{U}. If $J_1^2+J_2^2=1$ this is the case in $d=3$ for $J_1<-0.42$, $d=4$ for $J_1<-0.28$ and $d=5$ for $J_1<-0.22$.
\par
A similar theorem concerning nematic order (corresponding to correlation (b) in \ref{spincorrprop}) can be proved using the same methods. Unfortunately showing that one of the lower bounds obtained was positive proved difficult due to the seemingly unavoidable issue of bounding more complicated connection probabilities from below.

\subsection{Proof of Theorem \ref{thmneel}}

The result can be proved directly using the loop model. However in light of the proof of reflection positivity for the $J_2$ interaction \cite{L}, we can complete the proof more directly, appealing to Proposition \ref{spincorrprop} when required.
To begin we define a Hamiltonian
\begin{equation}
\tilde{H}^{J_1,J_2}_{\Lambda_L,\mathbf{h}}=\sum_{\{x,y\}\in\mathcal{E}}(J_1(S_x^1S_y^1-S_x^2S_y^2+S_x^3S_y^3)-J_2(S_x^1S_y^1-S_x^2S_y^2+S_x^3S_y^3)^2)+J_1\sum_{x\in\Lambda_L}h_xS_x^3.
\end{equation}
denote its Gibbs states $\langle\cdot\rangle^\sim$.
Note that for $\mathbf{h}=\mathbf{0}$ this Hamiltonian is unitarily equivalent to $H^{J_1,J_2}_{\Lambda_L}$ as can be seen from defining $U=\prod_{x\in\Lambda_B}e^{i\pi S_x^2}$ where $\Lambda_B\subset\Lambda_L$ is the odd sublattice. It has been shown \cite{D-L-S,L} that this Hamiltonian is reflection positive for $J_1\leq 0\leq J_2$. For external field $\mathbf{v}=(v_x)_{x\in\Lambda_L}$ for $v_x\in\mathbb{R}$ define
\begin{equation}
H(\mathbf{v})=\tilde{H}^{J_1,J_2}_{\Lambda_L,\mathbf{0}}+J_1\sum_{x\in\Lambda_L}(\Delta\mathbf{v})_xS_x^3
\end{equation}
and
\begin{equation}
Z(\mathbf{v})=Tr e^{-\beta H(\mathbf{v})}.
\end{equation}
By the usual methods for reflection positivity, using \cite{L} to handle the $J_2$ terms we can show that for any $\mathbf{v}$ we have $Z(\mathbf{v})\leq Z(\mathbf{0})e^{-J_1 \beta (\mathbf{v},\Delta\mathbf{v})}$. By choosing $\mathbf{v}$ according to $v_x=\eta\cos (k\cdot x)$ for $\eta$ small we can recover the infrared bound in \cite{D-L-S}:
\begin{equation}\label{infraredbound}
\widehat{(S_0^3,S_x^3)}_{Duh}(k)\leq\frac{1}{(-2J_1)\varepsilon(k)}
\end{equation}
where
\begin{equation}
(A,B)_{Duh}=\frac{1}{Z^{J_1,J_2}_{\Lambda_L,\beta}}\int_0^\beta\mathrm{d}s Tr A^* e^{-sH^{J_1,J_2}_{\Lambda_L}}Be^{-(\beta-s)H^{J_1,J_2}_{\Lambda_L}},
\end{equation}
is the Duhamel inner product. Now we use the Falk-Bruch inequality
\begin{equation}
\frac{1}{2}\langle A^*A+AA^*\rangle^{J_1,J_2}_{\Lambda_L,\beta}\leq \frac{1}{2}\sqrt{(A,A)_{Duh}}\sqrt{\langle[A^*,[H^{J_1,J_2}_{\Lambda_L},A]]\rangle^{J_1,J_2}_{\Lambda_L,\beta}}+\frac{1}{\beta}(A,A)_{Duh}.
\end{equation}
After some calculation and use of the SU(2) symmetry we find
\begin{equation}
\begin{aligned}
&\langle[\hat{S}_{-k}^3,[H^{J_1,J_2}_{\Lambda_L},\hat{S}_k^3]]\rangle^{J_1,J_2}_{\Lambda_L,\beta}=|\Lambda_L|\varepsilon(k+\pi)\bigg(2J_1\langle S_0^3S_{e_1}^3\rangle^{J_1,J_2}_{\Lambda_L,\beta}-J_2\big(4\langle(S_0^1)^2(S_{e_1}^3)^2\rangle^{J_1,J_2}_{\Lambda_L,\beta}
\\
&\qquad\qquad\qquad\quad-4\langle(S_0^3)^2(S_{e_1})^2\rangle^{J_1,J_2}_{\Lambda_L,\beta}-8\langle S_0^1S_0^3S_{e_1}^1S_{e_1}^3\rangle^{J_1,J_2}_{\Lambda_L,\beta}-4\langle S_0^1S_0^3S_{e_1}^3S_{e_1}^1\rangle^{J_1,J_2}_{\Lambda_L,\beta}\big)\bigg).
\end{aligned}
\end{equation}
Using Proposition \ref{spincorrprop} we can obtain the bound
\begin{equation}\label{doublecommbound}
\langle[\hat{S}_{-k}^3,[H^{J_1,J_2}_{\Lambda_L},\hat{S}_k^3]]\rangle^{J_1,J_2}_{\Lambda_L,\beta}\leq |\Lambda_L|\varepsilon(k+\pi)(-2J_1+8J_2)\mathbb{P}(\xoeasy\qquad\quad).
\end{equation}
Now we use the Fourier identity
\begin{equation}
\frac{1}{|\Lambda_L|}\sum_{x\in\Lambda_L}(-1)^{\|x\|}\langle S_0^3S_x^3\rangle^\sim=\langle S_0^3 S_y^3\rangle^\sim-\frac{1}{|\Lambda_L|}\sum_{k\in\Lambda_L^*\setminus\{0\}}e^{ik\cdot y}\widehat{\langle S_0^3S_x^3\rangle}\vspace{-0.1cm}^\sim(k)
\end{equation}
with $y=0$ and $y=e_1$. We see from using unitary operator $U$  with \eqref{infraredbound} and \eqref{doublecommbound} we have after taking limits the bounds in Theorem \ref{thmneel}.

\subsection*{Acknowledgments}
 
I am pleased to thank my supervisor Daniel Ueltschi for his support and useful discussions. I am also grateful to Graham Hobbs, Matthew Dunlop and David O'Conner for valuable discussions. This work is supported by EPSRC as part of the MASDOC DTC at the University of Warwick. Grant No. EP/HO23364/1. I would also like to thank the two referees for useful comments and observations.

\end{document}